\documentclass[11pt]{article}

\newcommand{\ifabs}[2]{#2}

\usepackage{comment}
\usepackage{cite}
\usepackage{fullpage}
\usepackage{color}
\usepackage{epic}
\usepackage{eepic}
\usepackage{epsfig}
\usepackage{xspace}
\usepackage{graphicx}
\usepackage{amsmath}
\usepackage{amsfonts}

\ifabs{ \excludecomment{fullproof}

\includecomment{abs}
\excludecomment{fullpaper}

}{

\excludecomment{proofsketch} \excludecomment{abs}
\includecomment{fullpaper}
}

\excludecomment{suppress}

\newcommand{\todo}[1]{\typeout{TODO: \the\inputlineno: #1}\textbf{[[[ #1 ]]]}}

\newcommand{\concept}[1]{\textbf{#1}}

\newtheorem{theorem}{Theorem}
\newtheorem{lemma}[theorem]{Lemma}
\newtheorem{corollary}[theorem]{Corollary}
\newtheorem{definition}[theorem]{Definition}

\newtheorem{Claim}[theorem]{Claim}
\newtheorem{proposition}[theorem]{Proposition}

\newcommand{\newloglike}[2]{\newcommand{#1}{\mathop{\rm #2}\nolimits}}
\newloglike{\sgn}{sgn}

\newcommand{\nul}[1]{{\it et al.\/}}

\newenvironment{proof}{\noindent{\bf Proof: }}{\nopagebreak\rule{1 ex}{0.8 em}\medskip}

\newcommand{\family}[1]{\mathcal{{#1}}}

\newcommand{\PM}{\mathrm{PM}}

\newcommand{\ANN}{\mathrm{ANN}}
\newcommand{\RO}{\mathrm{Butterfly\mbox{-}RO}}
\newcommand{\tabincell}[2]{\begin{tabular}{@{}#1@{}}#2\end{tabular}}

\begin{document}

\title{Certificates in Data Structures}
\author{Yaoyu Wang~\thanks{State Key Laboratory for Novel Software Technology, Nanjing University, China.}~\thanks{Email: \texttt{yaoyu.wang.nju@gmail.com}.}
\and Yitong Yin\footnotemark[1]~\thanks{Supported by NSFC grants 61272081, 61003023 and 61321491. Email: \texttt{yinyt@nju.edu.cn}.}  }

\date{}

\maketitle

\begin{abstract}
We study certificates in static data structures.
In the cell-probe model, certificates are the cell probes which can uniquely identify the answer to the query. As a natural notion of nondeterministic cell probes, lower bounds for  certificates in data structures immediately imply deterministic cell-probe lower bounds.
In spite of this extra power brought by  nondeterminism,
we prove that two widely used tools for cell-probe lower bounds: richness lemma of Miltersen~\textit{et al.}~\cite{miltersen1998data} and direct-sum richness lemma of P\v{a}tra\c{s}cu and Thorup~\cite{patrascu2010higher}, both hold for certificates in data structures with even better parameters.
Applying these lemmas and adopting existing reductions, we obtain certificate lower bounds for a variety of static data structure problems.
These certificate lower bounds are at least as good as the highest known cell-probe lower bounds for the respective problems.
In particular, for  approximate near neighbor (ANN) problem in Hamming distance,
our lower bound improves  the state of the art.
When the space is strictly linear, our lower bound for ANN in $d$-dimensional Hamming space becomes $t=\Omega(d)$, which along with the recent breakthrough for polynomial evaluation of Larsen~\cite{larsen2012higher}, are the only two $t=\Omega(d)$ lower bounds ever proved for any problems in the cell-probe model.

\end{abstract}

\section{Introduction}\label{section-intro}



In static data structure problems, a database is preprocessed to form a table according to certain encoding scheme, and upon each query to the database, an algorithm (decision tree) answers the query by adaptively probing the table cells. The complexity of this process is captured by the cell-probe model for static data structures. Solutions in this model are called cell-probing schemes.

The cell-probe model plays a central role in studying data structure lower bounds.
The existing cell-probe lower bounds for static data structure problems can be classified into the following three categories according to the techniques they use and the highest possible lower bounds supported by these techniques:
\begin{itemize}
\item Lower bounds implied by asymmetric communication complexity: Classic techniques introduced in the seminal work of Miltersen~\textit{et al.}~\cite{miltersen1998data} see a cell-probing scheme as a communication protocol between the query algorithm and the table, and the cell-probe lower bounds are implied by the asymmetric communication complexity lower bounds which are proved by the \emph{richness lemma} or round eliminations. In the usual setting that both query and data items are points from a $d$-dimensional space, the highest time lower bound that can be proved in this way is $t=\Omega\left(\frac{d}{\log s}\right)$ with a table of  $s$ cells. This bound is a barrier for the technique, because a matching upper bound can always be achieved by communication protocols.

\item Lower bounds proved by self-reduction using direct-sum properties:
The seminal works of P\v{a}tra\c{s}cu and Thorup~\cite{patrascu06pred,patrascu2010higher} introduce a very smart idea of many-to-one self-reductions, using which and by exploiting the direct-sum nature of problems, higher lower bounds can be proved for a near-linear space. The highest lower bounds that can be proved in this way is $t=\Omega\left({d}/{\log \frac{sw}{n}}\right)$ with a table of  $s$ cells each containing $w$ bits. Such lower bounds grow differently with near-linear space and polynomial space, which is indistinguishable in the communication model. 


\item Higher lower bounds for linear space: A recent breakthrough of Larsen~\cite{larsen2012higher} uses a technique refined from the cell sampling technique of Panigrahy \textit{et al.}~\cite{panigrahy2008geometric,panigrahy2010lower} to prove an even higher lower bound for the polynomial evaluation problem. This lower bound behaves as $t=\Omega(d)$ when the space is strictly linear. This separates for the first time between the cell-probe complexity with linear and near-linear spaces, and also achieves the highest cell-probe lower bound ever known for any data structure problems.

\end{itemize}

In this paper, we consider an even stronger model: \emph{certificates} in static data structures. A query to a database is said to have certificate of size $t$ if the answer to the query can be uniquely identified by the contents of $t$ cells in the table. This very natural notion represents the nondeterministic computation in cell-probe model and is certainly a lower bound to the complexity of deterministic cell-probing schemes.
This nondeterministic model has been explicitly considered before in a previous work~\cite{yin2010cell} of one of the authors of the current paper. 

Surprisingly, in spite of the seemingly extra power brought by the nondeterminism,
the highest cell-probe lower bound to date is in fact a certificate lower bound~\cite{larsen2012higher}.
Indeed, we  conjecture that for typical data structure problems, especially those hard problems, \emph{the complexity of certifying the answer should  dominate that of computing the answer}.\footnote{Interestingly, the only known exception to this conjecture is the predecessor search problem whose cell-probe complexity is a mild super-constant while the queries can be easily certified with constant cells in a sorted table.} This belief has been partially justified in~\cite{yin2010cell} by showing that a random static data structure problem is hard nondeterministically. 
In this paper, we further support this conjecture by showing that several mainstream techniques for cell-probe lower bounds in fact can imply as good or even higher certificate lower bounds.

\subsection{Our contributions}
We make the following contributions:
\begin{enumerate}
\item 
We prove a richness lemma for certificates in data structures, which improves the classic richness lemma for asymmetric communication complexity of Miltersen~\textit{et al.}~\cite{miltersen1998data} in two ways: (1) when applied to prove data structure lower bounds, our richness lemma implies lower bounds for a stronger \emph{nondeterministic} model; and (2) our richness lemma achieves better parameters than the classic richness lemma and may imply higher lower bounds.




\item 
We give a scheme for proving certificate lower bounds using a similar direct-sum based self-reduction of P\v{a}tra\c{s}cu and Thorup~\cite{patrascu2010higher}. 
The certificate lower bounds obtained from our scheme is at least as good as before when the space is near-linear or polynomial.
And for strictly linear space, our technique may support superior lower bounds, which was impossible for the direct-sum based techniques before. 

\item 
By applying these techniques, adopting the existing reductions, and modifying the reductions in the communication model to be model-independent, we prove certificate lower bounds for a variety of static data structure problems, listed in Table~\ref{table-results}. All these certificate lower bounds are at least as good as the highest known cell-probe lower bounds for the respective problems. And for approximate near neighbor (ANN), our $t=\Omega\left(d/\log\frac{sw}{nd}\right)$ lower bound improves the state of the art. When the space $sw=O(nd)$ is strictly linear, our lower bound for ANN becomes $t=\Omega(d)$, which along with the recent breakthrough for polynomial evaluation~\cite{larsen2012higher}, are the only two $t=\Omega(d)$ lower bounds ever proved for any problems in the cell-probe model.
\end{enumerate}

\begin{table}[tbt]
\begin{tabular}{|c|c|c|}
\hline
problem & \tabincell{c}{certificate lower bound\\ proved here} & \tabincell{c}{highest known\\ cell-probe lower bound}\\
\hline
bit-vector retrieval & $t=\Omega\left(\frac{m\log n}{\log s}\right)$& not known\\
\hline
lopsided set disjointness (LSD) & $t=\Omega\left(\frac{m\log n}{\log s}\right)$ & $t=\Omega\left(\frac{m\log n}{\log s}\right)$~\cite{miltersen1998data,patrascu06eps2n,patrascu11structures}\\
\hline
\tabincell{c}{ approximate near neighbor (ANN) \\ in Hamming space} & $t=\Omega\left(d/\log\frac{sw}{nd}\right)^\diamond$ & $t=\Omega\left(d/\log\frac{sw}{n}\right)^{\star}$~\cite{patrascu2010higher,panigrahy2008geometric} \\
\hline
partial match (PM) & $t=\Omega\left(d/\log\frac{sw}{n}\right)^\star$ & $t=\Omega\left(d/\log\frac{sw}{n}\right)^{\star}$~\cite{patrascu2010higher,panigrahy2008geometric} \\
\hline
3-ANN in $\ell_\infty$ & $t=\Omega\left(d/\log\frac{sw}{n}\right)^\star$ & $t=\Omega\left(d/\log\frac{sw}{n}\right)^{\star}$~\cite{patrascu2010higher} \\
\hline
\tabincell{c}{ reachability oracle \\ 2D stabbing \\ 4D range reporting } & $t=\Omega\left(\log n/\log\frac{sw}{n}\right)^\star$ & $t=\Omega\left(\log n/\log\frac{sw}{n}\right)^{\star}$~\cite{patrascu11structures} \\
\hline
2D range counting & $t=\Omega\left(\log n/\log\frac{sw}{n}\right)^\star$ & $t=\Omega\left(\log n/\log\frac{sw}{n}\right)^{\star}$~\cite{patrascu2007lower,patrascu11structures}\\
\hline
approximate distance oracle & $t=\Omega\left(\frac{\log n}{\alpha \log(s\log n/n)}\right)^\star$ & $t=\Omega\left(\frac{\log n}{\alpha \log(s\log n/n)}\right)^{\star}$~\cite{sommer2009distance}\\
\hline
\end{tabular}
\begin{tabular}{rl}
$\star$: & lower bound which grows differently with near-linear and polynomial space;\\
$\diamond$: & lower bound which grows differently with linear, near-linear, and polynomial space.
\end{tabular}
\caption{Certificate lower bounds proved in this paper.}\label{table-results}
\end{table}

\subsection{Related work}
The richness lemma, along with the round elimination lemma, for asymmetric communication complexity was introduced in~\cite{miltersen1998data}. The richness lemma was later widely used, for example in~\cite{borodin1999lower,barkol2000tighter,jayram2003cell,liu2004strong}, to prove lower bounds for high dimensional geometric problems, e.g.~nearest neighbor search. In~\cite{patrascu06eps2n,patrascu11structures}, a generalized version of richness lemma was proved to imply lower bounds for (Monte Carlo) randomized data structures.
A direct-sum richness theorem was first proved in the conference version of~\cite{patrascu2010higher}. Similar but less involved many-to-one reductions were used in~\cite{patrascu11structures} and~\cite{sommer2009distance} for proving lower bounds for certain graph oracles.



The idea of cell sampling was implicitly used in~\cite{patrascu2007lower} and independently in~\cite{panigrahy2008geometric}. This novel technique was later fully developed in~\cite{panigrahy2010lower} for high dimensional geometric problems and in~\cite{larsen2012higher} for polynomial evaluation. The lower bound in~\cite{larsen2012higher} actually holds for nondeterministic cell probes, i.e.~certificates.
The nondeterministic cell-probe complexity was studied for dynamic data structure problems in~\cite{husfeldt1998hardness} and for static data structure problems in\cite{yin2010cell}.

\section{Certificates in data structures}
A data structure problem is a function $f:X\times Y\rightarrow Z$ with two domains $X$ and $Y$. We call each $x\in X$ a \concept{query} and each $y\in Y$ a \concept{database}, and $f(x,y)\in Z$ specifies the result of query $x$ on database $y$. A code $T:Y\to\Sigma^s$ with an alphabet $\Sigma=\{0,1\}^w$ transforms each database $y\in Y$ to a \concept{table} $T_y=T(y)$ of  $s$ \concept{cells} each containing $w$ bits.
We use $[s]=\{1,2,\ldots,s\}$ to denote the set of indices of cells, and for each $i\in[s]$, we use $T_y(i)$ to denote the content of the $i$-th cell of table $T_y$.

A data structure problem is said to have \concept{$(s,w,t)$-certificates}, if  any database can be stored in a table of $s$ cells each containing $w$ bits, so that the result of each query can be uniquely determined by contents of at most $t$ cells. Formally, we have the following definition.
\begin{definition}\label{definition-certificate}
A data structure problem $f:X\times Y\rightarrow Z$ is said to have \concept{$(s,w,t)$-certificates},
if there exists a code $T:Y\rightarrow\Sigma^s$ with an alphabet $\Sigma=\{0,1\}^w$, such that for any query $x\in X$ and any database $y\in Y$, there exists a subset $P\subseteq[s]$ of cells with $|P|=t$, such that for any database $y'\in Y$, we have $f(x,y')=f(x,y)$ if $T_{y'}(i)=T_y(i)$ for all $i\in P$. 
\end{definition}

Because certificates represent nondeterministic computation in data structures, it is obvious that it has stronger computational power than cell-probing schemes.

\begin{proposition}\label{lemma-certificate-vs-probe}
For any data structure problem $f$, if there is a cell-probing scheme storing every database in $s$ cells each containing $w$ bits and answering every query within $t$ cell-probes, then $f$ has $(s,w,t)$-certificates.
\end{proposition}

\newcommand{\CertificateFormulation}{
Data structure certificates can be equivalently formulated as proof systems as well as certificates in decision trees of partial functions.

\paragraph*{As proof systems.}
In a previous work~\cite{yin2010cell}, an equivalent  formulation of data structure certificates as proof systems is used.
A data structure problem $f:X\times Y\rightarrow Z$ has $(s,w,t)$-certificates if and only if there exist a code $T:Y\rightarrow\Sigma^s$ with an alphabet $\Sigma=\{0,1\}^w$ and
a verifier $V:\{0,1\}^*\to Z\cup\{\bot\}$ where $\bot$ is a special symbol not in $Z$ indicating the failure of verification, so that for any query $x\in X$ and any database $y\in Y$, the followings are satisfied:
\begin{itemize}
\item Completeness: $\exists P\subseteq[s]$ with $|P|=t$ such that $V(x,\langle i,T_y(i)\rangle_{i\in P})=f(x,y)$;
\item Soundness: $\forall P'\subseteq[s]$ with $|P'|=t$, $V(x,\langle i,T_y(i)\rangle_{i\in P'})\in\{f(x,y),\bot\}$;
\end{itemize}
where $\langle i,T_y(i)\rangle_{i\in P}$ denotes the sequence of pairs $\langle i,T_y(i)\rangle$ for all $i\in P$.

\paragraph*{As certificates in decision trees.}
Certificate is a well-known notion is studies of decision tree complexity (see~\cite{buhrman2002complexity} for a survey).
A certificate in a Boolean function $h:\{0,1\}^n\to \{0,1\}$  for an input $x\in\{0,1\}^n$ is a subset $i_1,i_2,\ldots,i_t\in[n]$ of $t$ bits in $x$ such that for every $x'\in\{0,1\}^n$ satisfying that $x'(i_j)=x(i_j)$ for all $1\le j\le t$, it holds that $h(x)=h(x')$. And the certificate complexity of $h$, denoted by $C(h)$, is the minimum number of bits in a certificate in the worst-case of input $x$. The certificates and certificate complexity $C(h)$ can be naturally generalized to partial function $h:\Sigma^s\to Z$ with non-Boolean domain $\Sigma$ and range $Z$.

Given a data structure problem $f:X\times Y\rightarrow Z$, and a code $T:Y\rightarrow\Sigma^s$ with an alphabet $\Sigma=\{0,1\}^w$, for each query $x\in X$, the function $f$ can be naturally transformed into a partial function $f^T_x:\Sigma^s\to Z$ so that $f^T_x(T_y)=f(x,y)$ for every database $y\in Y$ and $f_x^T$ is not defined elsewhere. It is easy to verify that a data structure problem $f:X\times Y\rightarrow Z$ has $(s,w,t)$-certificates if and only if there exists a code $T:Y\rightarrow\Sigma^s$ with an alphabet $\Sigma=\{0,1\}^w$ such that $\max_{x\in X}C(f^T_x)\le t$, where $C(f^T_x)$ is the certificate complexity of the partial function $f^T_x:\Sigma^s\to Z$.
}

\ifabs{
In Appendix~\ref{appendix-certificate-formulation}, we state the equivalent formulations of
data structure certificates as proof systems and certificates in decision trees of partial functions.
}{
\CertificateFormulation
}





\section{The richness lemma}\label{section-characterization}
From now on, we focus on the decision problems where the output is either 0 or 1.
A data structure problem $f:X\times Y\rightarrow\{0,1\}$ can be naturally treated as an $|X|\times |Y|$ matrix whose rows are indexed by queries $x\in X$ and columns are indexed by data $y\in Y$.
The entry at the $x$-th row and $y$-th column is $f(x,y)$.
For $\xi\in\{0,1\}$, we say $f$ has a \concept{monochromatic $\xi$-rectangle} of size $k\times\ell$ if there is a combinatorial rectangle $A\times B$ with $A\subseteq X, B\subseteq Y, |A|=k$ and $|B|=\ell$, such that $f(x,y)=\xi$ for all $(x,y)\in A\times B$.
A matrix $f$ is said to be \concept{$(u,v)$-rich} if at least $v$ columns contain at least $u$ 1-entries.
The following richness lemma for cell-probing schemes is introduced in \cite{miltersen1998data}.
\begin{lemma}[Richness Lemma \cite{miltersen1998data}]\label{lemma-richness}
Let $f$ be a $(u,v)$-rich problem. If $f$ has an $(s,w,t)$-cell-probing scheme, then $f$ contains a monochromatic 1-rectangle of size $\frac{u}{2^{t\log s}}\times \frac{v}{2^{wt+t\log s}}$.
\end{lemma}
In \cite{miltersen1998data}, the richness lemma is proved for asymmetric communication protocols. A communication protocol between two parties Alice and Bob is called an $[A,B]$-protocol if Alice sends Bob at most $A$ bits and Bob sends Alice at most $B$ bits in total in the worst-case.
The richness lemma states that existence of $[A,B]$-protocol for a $(u,v)$-rich problem $f$ implies a submatrix of dimension $\frac{u}{2^{A}}\times \frac{v}{2^{A+B}}$ containing only 1-entries.
An $(s,w,t)$-cell-probing scheme can imply an $[A,B]$-protocol with $A=t\log s$ and $B=wt$, so the above richness lemma for the cell-probing schemes follows.

\subsection{Richness lemma for certificates}
We prove a richness result for data structure certificates, with even a better reliance on parameters. 
\begin{lemma}[Richness Lemma for data structure certificates]\label{lemma-richness-cert}
Let $f$ be a $(u,v)$-rich problem. If $f$ has $(s,w,t)$-certificates, then $f$ contains a monochromatic 1-rectangle of size $\frac{u}{\binom{s}{t}}\times \frac{v}{\binom{s}{t}2^{wt}}$.
\end{lemma}
\noindent\textbf{Remark.}
Note that we always have $\log {s\choose t}={t\log\frac{s}{t}+O(t)}\le {t\log s}$.
The bound in Lemma~\ref{lemma-richness-cert} is at least as good as the bound in classic richness lemma, even though now it is proved for nondeterministic computation. When $s$ and $t$ are close to each other, the bound in Lemma~\ref{lemma-richness-cert} is substantially better than that of classic richness lemma. Later in Section~\ref{section-direct-sum-property}, this extra gain is used in direct-sum reductions introduced in~\cite{patrascu2010higher} to achieve better time lower bounds for linear or near-linear space which match or improve state of the art. It is quite shocking to see all these achieved through a very basic reduction to the 1-probe case to be introduced later.

\paragraph{}
The classic richness lemma for asymmetric communication protocol is proved by a halving argument. Due to determinism of communication protocols (and cell-probing schemes), the combinatorial rectangle obtained from halving the universe are \emph{disjoint}. This disjointness no longer holds for the rectangles obtained from certificates because of nondeterminism. We resolve this issue by exploiting  combinatorial structures of rectangles obtained from data structure certificates.


The following preparation lemma is a generalization of the averaging principle.

\begin{lemma}\label{lemma-bad-data}
Let $\family{P}_1,\family{P}_2,\ldots,\family{P}_r\subset 2^{V}$ be partitions of $V$ satisfying $|\family{P}_i|\le k$ for every $1\le i\le k$.
There must exist a $y\in V$ such that $|\family{P}_i(y)|\ge\frac{|V|}{rk}$ for all $1\le i\le r$, where $\family{P}_i(y)$ denotes the partition block $B\in \family{P}_i$ containing $y$.
\end{lemma}
\begin{proof}
The lemma is proved by the probabilistic method. Let $y$ be uniformly chosen from $V$. Fix an arbitrary order of partition blocks for each partition $\family{P}_i$. Let $w_{ij}$ be the cardinality of the $j$-th block in $\family{P}_i$. Obviously the probability of $\family{P}_i(y)$ being the $j$-th block in $\family{P}_i$ is $\frac{w_{ij}}{|V|}$. By union bound, the probability that $|\family{P}_i(y)|<w$ is bounded by $\sum_{j:w_{ij}<w}\frac{w_{ij}}{|V|}<|\{j\,:\,w_{ij}<w\}|\frac{w}{|V|}$.
Since $|\family{P}_i|\le k$, for every $i$ there are at most $k$ many such $j$ satisfying that $w_{ij}<w$, thus
$\Pr\left[|\family{P}_i(y)|<\frac{|V|}{rk}\right]<k\cdot\frac{|V|/rk}{|V|}=\frac{1}{r}$.
Applying union bound again for all $\family{P}_i$, we have $\Pr\left[\exists 1\le i\le r, |\family{P}_i(y)|<\frac{|V|}{rk}\right]<1$, which means there exists a $y\in V$ such that $|\family{P}_i(y)|\ge\frac{|V|}{rk}$ for all $1\le i\le r$.
\end{proof}

We first prove the richness lemma for the 1-probe case.
\begin{lemma}\label{lemma-richness-one-cell}
Let $f$ be a $(u,v)$-rich problem. If $f$ has $(s,w,1)$-certificates, then $f$ contains a monochromatic 1-rectangle of size $\frac{u}{s}\times \frac{v}{s\cdot2^{w}}$.
\end{lemma}
\begin{proof}
Let $T: Y\to\Sigma^s$ where $\Sigma=\{0,1\}^w$ be the code in the $(s,w,1)$-certificates for $f$.
Let $V\subseteq Y$ denote the set of $v$ columns of $f$ that each contains at least $u$ 1-entries.
For each cell $1\le i\le s$, an equivalence relation $\sim_i$ on databases in $V$ can be naturally defined as follows: for any $y,y'\in V$, $y\sim_i y'$ if $T_y(i)=T_{y'}(i)$, that is, if databases $y$ and $y'$ look same in the $i$-th cell. Let $\family{P}_i$ denote the partition induced by the equivalence relation $\sim_i$. Each partition $\family{P}_i$ classifies the databases in $V$ according to the content of the $i$-th cell.
Obviously $|\family{P}_i|\le 2^w$, because the content of a cell can have at most $|\Sigma|=2^w$ possibilities, and we also have $\family{P}_i(y)=\{y'\in V\mid T_{y'}(i)=T_{y}(i)\}$ being the set of databases indistinguishable from $y$ by looking at the $i$-th cell, where $\family{P}_i(y)$ denotes the partition block $B\in \family{P}_i$ containing $y$. By Lemma~\ref{lemma-bad-data}, there always exists a bad database $y\in V$ such that $|\family{P}_i(y)|\ge\frac{|V|}{s\cdot2^{w}}=\frac{v}{s\cdot2^{w}}$ for all $1\le i\le s$.


For each database $y\in V$, let $X_{1}(y)=\{x\in X\mid f(x,y)=1\}$  denote the set of positive queries on database $y$, and for a subset $A\subseteq V$ of databases, let $X_1(A)=\bigcap_{y\in A}X_1(y)$ denote the set of queries which are positive on all databases in $A$. Note that $X_{1}(y)$ and $X_1(A)$ are the respective 1-preimages of Boolean functions $f(\cdot,y)$ and $\bigwedge_{y\in A} f(\cdot,y)$.
By definition, it is easy to see that $X_1(A)\times A$ is a monochromatic 1-rectangle for any $A\subseteq V$.

\textbf{Claim:} For any $y\in V$, it holds that $X_1(y)=\bigcup_{1\le i\le s} X_1(\family{P}_i(y))$.

It is easy to see the direction $\bigcup_{1\le i\le s} X_1(\family{P}_i(y))\subseteq X_1(y)$ holds because $X_1(A)\subseteq X_{1}(y)$ for any $A$ containing $y$ and clearly $y\in \family{P}_i(y)$. So we only need to prove the other direction. Since $f$ has $(s,w,1)$-certificates, for any positive query $x$ on database $y$ (i.e.~any $x\in X_1(y)$), there is a cell $i$ such that all databases $y'$ indistinguishable from $y$ by looking at the $i$-th cell (i.e.~all $y'\in\family{P}_i(y)$) answer the query $x$ positively (i.e.~$f(x,y')=f(x,y)=1$), which gives  $x\in X_1(\family{P}_i(y))$ by definition of $X_1(A)$.
This proves the direction $X_1(y) \subseteq \bigcup_{i\in [s]} X_1(\family{P}_i(y))$.



Consider the bad database $y\in V$ satisfying $|\family{P}_i(y)|\ge\frac{|V|}{s\cdot2^{w}}=\frac{v}{s\cdot2^{w}}$ for all $1\le i\le s$.
Due to the above claim, we have
\[
u\le \left|X_1(y)\right|=\left|\bigcup_{1\le i\le s} X_1(\family{P}_i(y))\right|\le\sum_{1\le i\le s}\left|X_1(\family{P}_i(y))\right|.
\]
By averaging principle, there exists a cell $i$ such that $\left|X_1(\family{P}_{i}(y))\right|\ge\frac{u}{s}$. This gives us a monochromatic 1-rectangle $X_1(\family{P}_{i}(y))\times \family{P}_{i}(y)$ of size  at least $\frac{u}{s}\times\frac{v}{s\cdot 2^{w}}$.
\end{proof}

The richness lemma for general case can be derived from the 1-probe case by a one-line reduction.
\begin{lemma}\label{lemma-certificate-reduction}
If a data structure problem $f$ has $(s,w,t)$-certificates, then $f$ has $\left(\binom{s}{t},w\cdot t,1\right)$-certificates.
\end{lemma}
\begin{proof}
Store every $t$-combination of cells with a new table of $\binom{s}{t}$ cells each of $w\cdot t$ bits.
\end{proof}


%

\newcommand{\SectionRichApp}{
We apply our richness lemma to two fundamental data structure problems: the bit-vector retrieval problem, and the lopsided set disjointness (LSD) problem. We prove certificate lower bounds matching the cell-probing scheme upper bounds, which shows that for these fundamental data structure problems, answering queries is as hard as certifying them.
\paragraph{Bit-vector retrieval.}
We consider the following fundamental problem: a database $y$ is a vector of $n$ bits, a query $x$ specifies $m$ indices, and the answer to the query returns the contents of these queried bits in the bit vector $y$.
Although is fundamental in database and information retrieval even judging by a glance, this problem has not been very well studied before (for a reason which we will see next).
We call this problem the  \concept{bit-vector retrieval} problem.
A naive solution is to explicitly store the bit-vector and access the queried bits directly, which gives an bit-probing scheme using $n$ bits and answering each query with $m$ bits.
A natural and important question is: can we substantially reduce the time cost by using a more sophisticated data structure with a tolerable overhead on space usage and allowing probing cells instead of bits? We shall see this is impossible in any realistic setting by showing a certificate lower bound.

We study a decision version of the bit-vector retrieval problem, namely the \concept{bit-vector testing} problem. Let $Y=\{0,1\}^n$ and  $X=[n]^m\times\{0,1\}^m$.  Each database $y\in Y$ is still an $n$-bit vector, and each query $x=(u,v)\in X$ consists of two parts: a tuple $u\in[n]^m$ of $m$ positions and a prediction $v\in\{0,1\}^m$ of the contents of these positions.
For $y\in\{0,1\}^n$ and $u\in[n]^m$, we use $y(u)$ to denote the $m$-tuple $(y(u_1),y(u_2),\ldots, y(u_m))$. The bit-vector testing problem $f:X\times Y\rightarrow\{0,1\}$ is then defined as that for any $x=(u,v)\in X$ and any $y\in Y$,  $f(x,y)$ indicates whether $y(u)=v$.

\begin{proposition}\label{proposition-bitvector-rectangle}
The bit-vector testing problem $f$ is $(n^m,2^n)$-rich and every $M\times N$ monochromatic 1-rectangles in $f$ must have $M\le (n-\log N)^m$.
\end{proposition}
\begin{proof}
We use the notation in the proof of Lemma~\ref{lemma-richness-one-cell}: we use $X_1(y)$ to denote set of positive queries on database $y$ and $X_1(A)$ to denote the set of queries positive on all databases in $A\subset Y$. Note that $X_1(y)$ contains all the rows at which column $y$ has 1-entries. 
It holds that $|Y|=n^m$ and for every $y\in Y$, we have $|X_1(y)|=|\{(u,v)\in [n]^m\times\{0,1\}^m\mid y(u)=v\}|=n^m$, thus $f$ is $(n^m, 2^n)$-rich. 

For any set $A\subseteq Y$, observe that $|X_1(A)|=|\{u\in[n]^m\mid \forall y, y'\in A, y(u)=y'(u)\}|$, i.e.~$|X_1(A)|$ is the number of such $m$-tuples of indices over which all bit-vectors in $A$ are identical.
Let $S$ denote the largest $S\subseteq[n]$ such that for every $i\in S$, $y(i)$ is identical for all $y\in A$. It is easy to see that $|X_1(A)|=|S|^m$ and $|A|\le2^{n-|S|}$, therefore it holds that $|X_1(A)|\le(n-\log|A|)^m$. Note that $X_1(A)\times A$ is precisely the maximal 1-rectangle with the set of columns $A$. Letting $N=|A|$, we prove that every $M\times N$ 1-rectangle must have $M\le(n-\log N)^m$.
\end{proof}

\begin{theorem}\label{lower-bound-bitvector}
If the bit-vector testing problem has $(s,w,t)$-certificates, then for any $0<\delta<1$, we have either $t\ge\frac{n^{1-\delta}}{w+\log s}$ or $t\ge\frac{\delta m\log n}{\log s}$.
\end{theorem}
\begin{proof}
Due to Proposition~\ref{proposition-bitvector-rectangle}, the problem is $(n^m,2^n)$-rich, and hence by Lemma~\ref{lemma-richness-cert}, if it has $(s,w,t)$-certificates, then it contains a 1-rectangle of size $\frac{n^m}{{s\choose t}}\times2^{n-wt-t\log {s\choose t}}$. As ${s\choose t}\leq s^t$, so we have a 1-rectangle of size $\frac{n^m}{s^t}\times2^{n-wt-t\log s}$, which by Proposition~\ref{proposition-bitvector-rectangle}, requires that $\frac{n^m}{s^t}\le (wt+t\log s)^m$.
For any $0<\delta<1$,
if $t<\frac{n^{1-\delta}}{w+\log s}$, then $t\ge\frac{\delta m\log n}{\log s}$.
\end{proof}

A standard setting for data structure is the lopsided case, where query is significantly shorter than database. For this case, the above theorem has the following corollary.
\begin{corollary}
Assuming $m=n^{o(1)}$, if the bit-vector testing problem has $(s,w,t)$-certificates for $w\le n^{1-\delta}$ where $\delta>0$ is an arbitrary constant, then $t=\Omega\left(\frac{m\log n}{\log s}\right)$.
\end{corollary}
With any polynomial space $s=n^{O(1)}$ and a wildly relaxed size of cell $n^{1-\delta}$, the above bound matches the naive solution of directly retrieving $m$ bits, implying that the fundamental problem of retrieving part of a bit vector cannot be made any easier in a general setting, because queries are hard to certify.


\paragraph{Lopsided set disjointness.}
The set disjointness problem plays a central role in communication complexity and complexity of data structures.
Assuming a data universe $[N]$, the input domains are $X={[N]\choose m}$ and $Y={[N]\choose n}$ where $m\le n<\frac{N}{2}$. For each query set $x\in X$ and data set $y\in Y$, the set disjointness problem $f(x,y)$ returns a bit indicating the emptyness of $x\cap y$. 
The following proposition is implicit in~\cite{miltersen1998data}.
\begin{proposition}[Milersen~\textit{et al.}~\cite{miltersen1998data}]\label{proposition-set-disjoint}
The set disjointness problem $f$ is $\left({N-n\choose m},{N\choose n}\right)$-rich, and for every $n\le u\le N$, any monochromatic 1-rectangle in $f$ of size $M\times {u\choose n}$ must have $M\le {N-u\choose m}$.
\end{proposition}
\begin{proof}
We use the notation in the proof of Lemma~\ref{lemma-richness-one-cell}: let $X_1(y)$ denote set of positive queries on database $y$ and $X_1(A)$ denote the set of queries positive on all databases in $A\subset Y$. $X_1(y)$ contains all the rows at which column $y$ has 1-entries.
It holds that $|Y|={N\choose n}$ and for every set $y\in Y$ with $|y|=n$, we have $|X_1(y)|=|\{x\mid x\subset [N], |x|=m, x\cap y=\emptyset\}|={N-n\choose m}$, thus $f$ is $\left({N-n\choose m}, {N\choose n}\right)$-rich. 

For any set $A\subseteq Y$ with $|A|={u\choose n}$, let $y'=\bigcup_{y\in A}y$. We have $|y'|\geq u$. For $X_1(A)=\{x\mid \forall y\in Y, x\cap y=\emptyset \}$, we have $|X_1(A)|\leq {N-|y'|\choose m}\leq{N-u\choose m}$.
Thus we get the conclusion.
\end{proof}

Applying the above proposition and Lemma~\ref{lemma-richness-cert}, we have the following certificate lower bound.
\begin{theorem}\label{theorem-set-disjoint}
If the set disjointness problem has $(s,w,t)$-certificates, then for any $0<\delta<1$, we have either $t\ge\frac{n^{1-\delta}}{w+\log s}$ or $t\ge\frac{\delta m(\log n-o(1)))}{\log s}$.
\end{theorem}
\begin{proof}
Due to Proposition~\ref{proposition-set-disjoint}, the problem is $\left({N-n\choose m},{N\choose n}\right)$-rich, and hence by Lemma~\ref{lemma-richness-cert}, if it has $(s,w,t)$-certificates, then it contains a 1-rectangle of size $\frac{{N-n\choose m}}{{s\choose t}}\times\frac{{N\choose n}}{{s\choose t}2^{wt}}$. As ${s\choose t}\leq s^t$, so we have a 1-rectangle of size ${N-n\choose m}/{2^{t\log s}}\times{N\choose n}/{2^{wt+t\log s}}$. Let $a=t\log s$ and $b=wt$. Let $M,u$ denote the parameters in Proposition~\ref{proposition-set-disjoint} respectively, so $M={N-n\choose m}/{2^a}$ and ${u\choose n}={N\choose n}/{2^{a+b}}$. Let $k=(N-n)/2^{a/m}$. Since ${k\choose m}\leq{N-n\choose m}/2^{a}\leq{N-u\choose m}$ by Proposition~\ref{proposition-set-disjoint}, we have $k\leq N-u$, which leads to $u\leq N-k$. Now we have ${N\choose n}/2^{a+b}={u \choose n}\leq {N-k \choose n}$ and therefore $2^{a+b}\geq{N\choose n}/{N-k\choose n}>(\frac{N}{N-k})^n>(1+k/N)^n=(1+(N-n)/2^{a/m}N)^n\geq(1+2^{-a/m-1})^n$. By taking logarithm, we have $a+b\geq n\log(1+2^{-a/m-1})>n\cdot 2^{-a/m-1}$. If $a+b< n^{1-\delta}$ for any $\delta$, then $2^{a/m+1}\geq n^{\delta}$, thus $a\geq \delta m (\log n-o(1))$. Replacing $a,b$ with $t\log s, wt$ respectively, we get the conclusion.
\end{proof}

This certificate lower bound matches the well-known cell-probe lower bound for set-disjointness~\cite{miltersen1998data,patrascu11structures}.
The most interesting case of the problem is the lopsided case where $m=n^{o(1)}$. 
A calculation gives us the following corollary.
\begin{corollary}
Assume $m=n^{o(1)}$ and  $\alpha n\le N\le n^c$ for arbitrary constants $\alpha, c>1$. If the set disjointness problem has $(s,w,t)$-certificates for $w\le n^{1-\delta}$ where $\delta>0$ is an arbitrary constant, then $t=\Omega\left(\frac{m\log n}{\log s}\right)$.
\end{corollary}

}

\ifabs{
In Appendix~\ref{appendix-richness-app}, we apply Lemma~\ref{lemma-richness-cert} to prove certificate lower bounds for bit-vector retrieval problem and lopsided set disjointness problem.
}
{
\subsection{Applications}\label{section-certificate-lower-bounds}
\SectionRichApp
}

\section{Direct-sum richness lemma}\label{section-direct-sum-property}
In this section, we prove a richness lemma for certificates using direct-sum property of data structure problems. Such a lemma was introduced in~\cite{patrascu2010higher} for cell-probing schemes, which is used to prove some highest known cell-probe lower bounds with \emph{near-linear spaces}. 

Consider a vector of problems $\bar{f}=(f_1,\dots,f_k)$ where every $f_i:X\times Y\rightarrow \{0,1\}$ is defined on the same domain $X\times Y$.
Let $\bigoplus^{k}\bar{f}:([k]\times X)\times Y^k\rightarrow\{0,1\}$ be a problem defined as follows: $\bigoplus^{k}\bar{f}((i,x),\bar{y})=f_i(x,y_i)$ for every $(i,x)\in[k]\times X$  and every $\bar{y}=(y_1,y_2,\ldots,y_k)\in Y^k$. In particular, for a problem $f$ we denote $\bigoplus^{k}f=\bigoplus^{k}\bar{f}$  where $\bar{f}$ is a tuple of $k$ copies of problem $f$.


\begin{lemma}[direct-sum richness lemma for certificates]\label{lemma-direct-sum-cert}
Let $\bar{f}=(f_1,f_2\dots\,f_k)$ be a vector of problems 
such that for each $i=1,2,\ldots,k$, we have $f_i:X{\times}Y\rightarrow\{0,1\}$ and $f_i$ is $(\alpha|X|,\beta|Y|)$-rich. If problem $\bigoplus^{k}\bar{f}$ has $(s,w,t)$-certificates for a $t\leq\frac{s}{k}$, then there exists a $1\le i\le k$ such that $f_i$ contains a monochromatic 1-rectangle of size $\frac{{\alpha}^{O(1)}|X|}{2^{O(t\log {\frac{s}{kt}})}}\times \frac{{\beta}^{O(1)}|Y|}{2^{O(wt+t\log {\frac{s}{kt}})}}$.
\end{lemma}

\paragraph{Remark 1.} The direct-sum richness lemma proved in~\cite{patrascu2010higher} is for asymmetric communication protocols as well as cell-probing schemes, and gives a rectangle size of $\frac{{\alpha}^{O(1)}|X|}{2^{O(t\log {\frac{s}{k}})}}\times \frac{{\beta}^{O(1)}|Y|}{2^{O(wt+t\log {\frac{s}{k}})}}$. Our direct-sum richness lemma has a better rectangle bound. This improvement may support stronger lower bounds which separate between linear and near-linear spaces.

\paragraph{Remark 2.}
A key idea to apply this direct sum based lower bound scheme is to exploit the extra power gained by the model from solving $k$ problem instances in parallel. In~\cite{patrascu2010higher}, this is achieved by seeing cell probes as communications between query algorithm and table, and $t$-round adaptive cell probes for answering $k$ parallel queries can be expressed in $t\log {s\choose k}$ bits instead of naively $kt\log s$ bits. For our direct-sum richness lemma for certificates, in contrast, we will see (in Lemma~\ref{lemma-direct-sum}) that unlike communications, the parallel simulation of certificates does not give us any extra gain, however, in our case all extra gains are provided by the improved bound in Lemma~\ref{lemma-richness-cert}, the richness lemma for certificates. Indeed, all our extra gains by ``parallelism'' are offered by the one-line reduction in Lemma~\ref{lemma-certificate-reduction}, which basically says that the certificates for $k$ instances of a problem can be expressed in $\log{s\choose kt}$ bits, even better than the $t\log {s\choose k}$-bit bound for communications. Giving up  adaptivity is essential to this improvement on the power of parallelism, so that all $kt$ cells can be chosen at once which gives the $\log{s\choose kt}$-bit bound: \emph{we are now not even parallel over instances, but also parallel over time}.


\paragraph*{}
The idea of proving Lemma~\ref{lemma-direct-sum-cert} can be concluded as: (1) reducing the problem $\bigoplus^{k}\bar{f}$ from a direct-product problem $\bigwedge^{k}\bar{f}$ whose richness and monochromatic rectangles can be easily translated between $\bigwedge^{k}\bar{f}$ and subproblems $f_i$; and (2) applying Lemma~\ref{lemma-richness-cert}, the richness lemma for certificates, to obtain large monochromatic rectangles for the direct-product problem.


We first define a direct-product operation on vector of problems. For $\bar{f}=(f_1,\dots,f_k)$ with $f_i:X\times Y\rightarrow \{0,1\}$ for every $1\le i\le k$, let $\bigwedge^{k}\bar{f}:X^k\times Y^k\rightarrow \{0,1\}$ be a direct-product problem defined as: $\bigwedge^{k}\bar{f}(\bar{x},\bar{y})=\prod_i{f_i(x_i,y_i)}$ for every $\bar{x}=(x_1,\dots,x_k)$ and every $\bar{y}=(y_1,\dots,y_k)$.

\begin{lemma}\label{lemma-direct-sum}
For any $\bar{f}=(f_1,\dots,f_k)$, if $\bigoplus^{k}\bar{f}$ has $(s,w,t)$-certificates for a $t\leq\frac{s}{k}$, then $\bigwedge^{k}\bar{f}$ has $(s,w,kt)$-certificates.
\end{lemma}
\begin{proof}
Suppose that $T:Y^k\to \Sigma^s$ with $\Sigma=\{0,1\}^w$ is the code used to encode databases to tables in the $(s,w,t)$-certificates of $\bigoplus^{k}\bar{f}$.
For problem $\bigwedge^{k}\bar{f}$, we use the same code $T$ to prepare table. And for each input $(\bar{x},\bar{y})$ of problem $\bigwedge^{k}\bar{f}$ where $\bar{x}=(x_1,\dots,x_k)$ and $\bar{y}=(y_1,\dots,y_k)$, suppose that for each $1\le i\le k$, $P_i\subset[s]$ with $|P_i|=t$ is the set of $t$ cells in table $T_{\bar{y}}$ to uniquely identify the value of $\bigoplus^{k}\bar{f}((i,x_i),\bar{y})$, then let $P=P_1\cup P_2\cup\cdots\cup P_k$ so that $|P|\le kt$. It is easy to verify that the set $P$ of at most $kt$ cells in $T_{\bar{y}}$ uniquely identifies the value of $\bigwedge^{k}\bar{f}(\bar{x},\bar{y})=\bigwedge_{1\le i\le k}\left(\bigoplus^{k}\bar{f}((i,x_i),\bar{y})\right)$ because it contains all cells which can uniquely identify the value of $\bigoplus^{k}\bar{f}((i,x_i),\bar{y})$ for every $1\le i\le k$. Therefore, problem $\bigwedge^{k}\bar{f}$ has $(s,w,kt)$-certificates.
\end{proof}


The following two lemmas are from~\cite{patrascu2010higher}. These lemmas give  easy translations of richness and  monochromatic rectangles between the direct-product problem $\bigwedge^{k}\bar{f}$ and subproblems $f_i$.

\begin{lemma}[P\v{a}tra\c{s}cu and Thorup \cite{patrascu2010higher}]\label{lemma-direct-sum-richness}
If $\bar{f}=(f_1,f_2\dots\,f_k)$ has $f_i:X{\times}Y\rightarrow\{0,1\}$ and $f_i$ is  $(\alpha|X|,\beta|Y|)$-rich for every $1\le i\le k$, then $\bigwedge^{k}\bar{f}$ is $((\alpha|X|)^k,(\beta|Y|)^k)$-rich.
\end{lemma}
\begin{lemma}[P\v{a}tra\c{s}cu and Thorup \cite{patrascu2010higher}]\label{lemma-direct-sum-rectangle}
For any $\bar{f}=(f_1,\dots,f_k)$ with $f_i:X{\times}Y\rightarrow\{0,1\}$ for every $1\le i\le k$, if $\bigwedge^{k}\bar{f}$ contains a monochromatic 1-rectangle of size $(\alpha|X|)^k\times(\beta|Y|)^k$, then there exists a $1\le i \le k$ such that $f_i$ contains a monochromatic 1-rectangle of size $(\alpha)^3|X|\times(\beta)^3|Y|$.
\end{lemma}


The direct-sum richness lemma can be easily proved by combining the above lemmas with the richness lemma for certificates.

\begin{proof}[Proof of Lemma~\ref{lemma-direct-sum-cert}]
If $\bigoplus^{k}\bar{f}$ has $(s,w,t)$-certificates, then by Lemma~\ref{lemma-direct-sum}, the direct-product problem $\bigwedge^{k}\bar{f}$ has $(s,w,kt)$-certificates. Since every $f_i$ in $\bar{f}=(f_1,f_2,\ldots,f_k)$ is $(\alpha|X|,\beta|Y|)$-rich, by Lemma \ref{lemma-direct-sum-richness} we have  that $\bigwedge^{k}\bar{f}$ is $((\alpha|X|)^k,(\beta|Y|)^k)$-rich. Applying Lemma~\ref{lemma-richness-cert}, the richness lemma for certificates, problem $\bigwedge^{k}\bar{f}$ has a 1-rectangle of size $\frac{(\alpha|X|)^k}{\binom{s}{kt}}\times\frac{(\beta|Y|)^k}{\binom{s}{kt}2^{kwt}}$. Then due to Lemma~\ref{lemma-direct-sum-rectangle}, we have a problem $f_i$ who contains a monochromatic 1-rectangle of size $\frac{{\alpha}^{O(1)}|X|}{2^{O(t\log {\frac{s}{kt}})}}\times \frac{{\beta}^{O(1)}|Y|}{2^{O(wt+t\log {\frac{s}{kt}})}}$.
\end{proof}
\subsection{Applications}\label{section-direct-sum-application}
We then apply the direct-sum richness lemma to prove lower bounds for two important high dimensional problems: approximate near neighbor (ANN) in hamming space and partial match (PM).
\begin{itemize}
\item
For ANN in $d$-dimensional hamming space, we prove a $t=\Omega(d/\log\frac{sw}{nd})$ lower bound for $(s,w,t)$-certificates.
The highest known cell-probing scheme lower bound for the problem is  $t=\Omega(d/\log\frac{sw}{n})$.
In a super-linear space, our certificate lower bound matches the highest known lower bound for cell-probing scheme; and for linear space, our lower bound becomes $t=\Omega(d)$, which gives a strict improvement, and also matches the highest cell-probe lower bound ever known for any problem (which has only been achieved for polynomial evaluation~\cite{larsen2012higher}).
\item
For $d$-dimensional PM, we prove a $t=\Omega(d/\log\frac{sw}{n})$ lower bound for $(s,w,t)$-certificates, which matches the highest known cell-probing scheme lower bound for the problem in~\cite{patrascu2010higher}.
\end{itemize}


\subsubsection{Approximate near neighbor (ANN)}
The near neighbor problem $\mathrm{NN}_n^d$ in a $d$-dimensional metric space is defined as follows: a database $y$ contains $n$ points from a $d$-dimensional metric space, for any query point $x$ from the same space and a distance threshold $\lambda$, the problem asks whether there is a point in database $y$ within distance $\lambda$ from $x$. The approximate near neighbor problem $\ANN_{n}^{\lambda,\gamma,d}$ is similarly defined, except upon a query $x$ to a database $y$, answering ``yes'' if there is a point in database $y$ within distance $\lambda$ from $x$ and ``no'' if all points in $y$ are $\gamma\lambda$-far away from $x$ (and answering arbitrarily if otherwise).


We first prove a lower bound for $\ANN_{n}^{\lambda,\gamma,d}$ in Hamming space $X=\{0,1\}^d$, where for any two points $x,x'\in X$ the distance between them is given by Hamming distance $h(x,x')$.

\begin{suppress}
\begin{theorem}\label{theorem-ANN}
For problem $\ANN_{n}^{\lambda,\gamma,d}$ in Hamming cube $\{0,1\}^d$, assume $d\geq(1+5\gamma)\log n$. There exists some $\lambda$, such that if $ANN_{n}^{\lambda,\gamma,d}$ has $(s,w,t)$-certificates, where $w=d^{O(1)}$, then $t=\Omega({{\frac{d}{\gamma^3}}/{\log\frac{sd}{n}}})$.
\end{theorem}
\begin{proof}
Assume problem $ANN_{n}^{\lambda,\gamma,d}$ has $(s,w,t)$-certificates. Let $D=d/(1+5\gamma)\geq\log n$ and $k=n/N$ for $N<n$ to be chosen later. In [Theorem 11,Ref?], Mihai has proved a solution for $ANN_{n}^{\lambda,\gamma,d}$ can also work as a solution for $\bigoplus^k ANN_N^{\lambda,\gamma,D}$ with the same complexity. It's easy to see this also holds for certificates. So $\bigoplus^k ANN_N^{\lambda,\gamma,D}$ also has $(s,w,t)$-certificates.

By [Ding Liu], there exists a $\lambda$ such that:
\begin{itemize}
\item by [Ding Liu, Claim 10],$ANN_N^{\lambda,\gamma,D}$ is $[2^{D-1},2^{ND}]$-rich.
\item by [Ding Liu, Claim 11],$ANN_N^{\lambda,\gamma,D}$ doesn't have a 1-rectangle of size $2^{D-D/(169\gamma^2)}\times 2^{ND-ND/(32\gamma^2)}$.
\end{itemize}
Note the query domain $|X|=2^D$, and the data domain $|Y|=2^{ND}$, so the problem is $[|X|/2,|Y|]$-rich. By Lemma \ref{lemma-direct-sum-cert}, we have either $t=\Omega(\frac{D}{\gamma^2}/\log\frac{s}{k})$ or $t=\Omega(\frac{ND}{\gamma^2}/w)$. Fix $N=w=d^{O(1)}$, the first item is the smaller one. Now we get $t=\Omega(\frac{d}{\gamma^3}/\log\frac{sN}{n})=\Omega(\frac{d}{\gamma^3}/\log\frac{sd}{n})$.
\end{proof}
\end{suppress}

The richness and monochromatic rectangles of $\ANN_n^{\lambda,\gamma,n}$ were analyzed in~\cite{liu2004strong}.

\begin{Claim}[Claim 10 and 11 in~\cite{liu2004strong}]\label{claim-rectangle-ANN}
There is a $\lambda\leq d$ such that $\ANN_n^{\lambda,\gamma,d}$ is $(2^{d-1},2^{nd})$-rich and $\ANN_n^{\lambda,\gamma,d}$ does not contain a 1-rectangle of size $2^{d-d/(169\gamma^2)}\times 2^{nd-nd/(32\gamma^2)}$.
\end{Claim}

A model-independent self-reduction of ANN was constructed in~\cite{patrascu2010higher}.

\begin{Claim}[Theorem 6 in \cite{patrascu2010higher}]\label{claim-selfreduction-ANN}
For $D=d/(1+5\gamma)\geq\log n$, $N<n$ and $k=n/N$, there exist two functions $\phi_X,\phi_Y$ such that $\phi_X$ (and $\phi_Y$) maps each query $(x,i)$ (and database $\bar{y}$) of $\bigoplus^k \ANN_N^{\lambda,\gamma,D}$ to a query $x'$ (and database $y'$) of $\ANN_{n}^{\lambda,\gamma,d}$ and it holds that $\bigoplus^k \ANN_N^{\lambda,\gamma,D}((x,i),\bar{y})= \ANN_{n}^{\lambda,\gamma,d}(x',y')$.
\end{Claim}

We then prove the following certificate lower bound for ANN.

\begin{theorem}\label{theorem-ANN-v2}
For $\ANN_{n}^{\lambda,\gamma,d}$ in $d$-dimensional Hamming space, assuming $d\geq(1+5\gamma)\log n$, there exists a $\lambda$, such that if $\ANN_{n}^{\lambda,\gamma,d}$ has $(s,w,t)$-certificates, then $t=\Omega\left({{\frac{d}{\gamma^3}}/{\log\frac{sw\gamma^3}{nd}}}\right)$.
\end{theorem}
\begin{proof}
Due to the model-independent reduction from $\bigoplus^k \ANN_N^{\lambda,\gamma,D}$ to $\ANN_{n}^{\lambda,\gamma,d}$ of Claim~\ref{claim-selfreduction-ANN}, existence of $(s,w,t)$-certificates for $\ANN_{n}^{\lambda,\gamma,d}$ implies the existence of $(s,w,t)$-certificates for $\bigoplus^k \ANN_N^{\lambda,\gamma,D}$.


Note that for problem $\ANN_N^{\lambda,\gamma,D}$, the size of query domain  is $|X|=2^D$, and the size of data domain is $|Y|=2^{ND}$, so applying Claim~\ref{claim-rectangle-ANN}, the problem is $(|X|/2,|Y|)$-rich. Assuming that $t\le\frac{s}{k}$, by Lemma \ref{lemma-direct-sum-cert}, $\ANN_N^{\lambda,\gamma,D}$ contains a 1-rectangle of size $2^{D}/2^{O(t\log\frac{s}{kt})}\times 2^{ND}/2^{O(wt+t\log\frac{s}{kt})}$. Due to Claim~\ref{claim-rectangle-ANN}, and by a calculation,
we have either $t=\Omega\left(\frac{D}{\gamma^2}/\log\frac{s}{kt}\right)$ or $t=\Omega\left(\frac{ND}{\gamma^2}/w\right)$. We then choose $N=w$. Note that such choice of $N$ may violate the assumption $t\le\frac{s}{k}$ (that is, $N\ge \frac{tn}{s}$) only when it implies an even higher lower bound $t>\frac{sw}{n}$.
With this choice of $N=w$, the bound $t=\Omega\left(\frac{D}{\gamma^2}/\log\frac{s}{kt}\right)$ is the smaller one in the two branches.
Substituting $D=d/(1+5\gamma)$ and $k=n/N$ we have $t=\Omega\left(\frac{d}{\gamma^3}/\log\frac{sN}{nt}\right)=\Omega\left(\frac{d}{\gamma^3}/\log\frac{sw}{nt}\right)$.
Multiplying both side by a $\Delta=\frac{sw}{nd}$ 
gives us $\Delta\cdot\gamma^3=\Omega\left(\frac{\Delta d}{t}/\log\frac{\Delta d}{t}\right)$. Assuming $\Delta'=\frac{\Delta d}{t}$, we have $\frac{\Delta '}{\log \Delta'}=O(\Delta \gamma^3)$. The function $f(x)=\frac{x}{\log x}$ is increasing for $x>1$, so we have $\Delta'=O(\Delta\gamma^3\log(\Delta \gamma^3))$, which gives us the lower bound $t=\Omega\left({{\frac{d}{\gamma^3}}/{\log\frac{sw\gamma^3}{nd}}}\right)$.
\end{proof}



For general space, when points are still from the Hamming cube $\{0,1\}^d$, for any two points $x,x'\in\{0,1\}^d$, the Hamming distance $h(x,x')=\|x-x'\|_1=\|x-x'\|_2^2$. And by setting $\gamma=1$, we have the following corollary for exact near neighbor.

\begin{corollary}\label{corollary-NN}
There exists a constant $C$ such that for problem $\mathrm{NN}_{n}^{d}$ with Hamming distance, Manhattan norm $\ell_1$ or Euclidean norm $\ell_2$, assuming $d\geq{C\log n}$, if $\mathrm{NN}_{n}^{d}$ has $(s,w,t)$-certificates, then $t=\Omega({{d}/{\log\frac{sw}{nd}}})$.
\end{corollary}

\subsubsection{Partial match}
The partial match problem is another fundamental high-dimensional problem.
The $d$-dimensional partial match problem $\PM_n^d$ is defined as follows: a database $y$ contains $n$ strings from $\{0,1\}^d$, for any query pattern $x\in\{0,1,*\}^d$, the problem asks whether there is a string $z$ in database $y$ matching pattern $x$, in such a way that $x_i=z_i$ for all $i\in[d]$ that $x_i\neq *$.

\begin{theorem}\label{theorem-PM}
Assuming $d\geq{2\log n}$, if problem $\PM_{n}^{d}$ has $(s,w,t)$-certificates for a $w=d^{O(1)}$, then $t=\Omega\left({{d}/{\log\frac{sd}{n}}}\right)$.
\end{theorem}
\newcommand{\ProofPM}{
The proof is almost exactly the same as the proof of partial match lower bound in~\cite{patrascu2010higher}. We restate the proof in the context of certificates.
Let  $N=n/k$ and $D=d-\log k\geq d/2$.
We have the following model-independent reduction from $\bigoplus^k \PM_N^D$ to $\PM_{n}^{d}$: For the data input $\bigoplus^k \PM_N^D$, we  add the subproblem index in binary code, which takes $\log k$ bits,  as a prefix for every string. And for the query, we also add  the subproblem index $i$ in binary code as a prefix to the query pattern to form a new query in $\PM_{n}^{d}$. It is easy to see $\PM_{n}^{d}$ solves $\bigoplus^k \PM_N^D$ with such a reduction, and $(s,w,t)$-certificates for $\PM_{n}^{d}$ are $(s,w,t)$-certificates for $\bigoplus^k \PM_N^D$.

In Theorem~11 of~\cite{patrascu2010higher}, it is proved that on a certain domain $X\times Y$ for $\PM_N^D$:
\begin{itemize}
\item  $\PM_N^D$ is $\left({|X|}/{4},{|Y|}/{4}\right)$-rich. In fact, in~\cite{patrascu2010higher} it is only proved that the density of 1s in $\PM_N^D$ is at least $1/2$, which easily implies the richness due to an averaging argument.
\item $\PM_N^D$ has no 1-rectangle of size $|X|/2^{O(D)}\times |Y|/2^{O(\sqrt{N}/D^2)}$.
\end{itemize}

Assuming that $t\le\frac{s}{k}$, by Lemma~\ref{lemma-direct-sum-cert}, we have either $t\log\frac{s}{k}=\Omega(D)$ or $t\log\frac{s}{k}+wt=\Omega(\sqrt N/D^2)$.  We choose $N=w^2\cdot D^8$. Note that this choice of $N$ may violate the assumption $t\le\frac{s}{k}$ only when an even higher lower bound $t>\frac{sw^2D^8}{n}=\Omega(d^2)$ holds. With this choice of $N=w^2\cdot D^8=d^{O(1)}$, the second bound above becomes $t=\Omega(d^2)$, while the first becomes $t=\Omega\left(d/\log\frac{sd}{nt}\right)=\Omega\left(d/\log\frac{sd}{n}\right)$.
}
\ifabs{
The proof of this theorem is in Appendix~\ref{appendix-direct-sum}.
}{
\begin{proof}
\ProofPM
\end{proof}
}

\begin{suppress}
In \cite{p2011unifying}, Patrascu shows that BLOCKED-LSD on universe $[U\cdot B]$ reduces to partial match over $N=U\cdot B$ strings in dimension $D=O(U\log B)$. We have proved in Section \ref{section-characterization} that BLOCKED-LSD on universe $[U\cdot B]$ has lower bound $t\log \frac{s}{t}=\Omega(U\cdot \log B)$ or $tw =\Omega(U\cdot \sqrt B)$. By Lemma \ref{lemma-direct-sum-cert}, we have either $t\log\frac{s}{kt}=\Omega(U\log B)=\Omega(D)$ or $wt=\Omega(U\cdot \sqrt B)=\Omega(\sqrt N)$.
 Now set $N=w^2\cdot D^4=d^{O(1)}$, the second inequality becomes $t=\Omega(d^2)$, while the first becomes $t=\Omega(d/\log\frac{sd}{nt})=\Omega(d/\log\frac{sd}{n})$. We thus conclude that $t=\Omega(d/\log\frac{sd}{n})$.

 In [Unifying Landscape], Patrascu shows that BLOCKED-LSD on universe $[U\cdot B]$ reduces to partial match over $N=U\cdot B$ strings in dimension $D=O(U\log B)$. We have proved in Section \ref{section-characterization} that BLOCKED-LSD on universe $[U\cdot B]$ has lower bound $t\log \frac{s}{t}=\Omega(U\cdot \log B)$ or $tw =\Omega(U\cdot \sqrt B)$. By Lemma \ref{lemma-direct-sum-cert}, we have either $t\log\frac{s}{kt}=\Omega(U\log B)=\Omega(D)$ or $wt=\Omega(U\cdot \sqrt B)=\Omega(\sqrt N)$.
 Now set $N=w^2\cdot D^4=d^{O(1)}$, the second inequality becomes $t=\Omega(d^2)$, while the first becomes $t=\Omega(d/\log\frac{sd}{nt})=\Omega(d/\log\frac{sd}{n})$. We thus conclude that $t=\Omega(d/\log\frac{sd}{n})$.
 \end{suppress}

\begin{suppress}
 In [Borodin, Ostrovsky, and Rabani], they observe that partial match can be reduced to near neighbor problem in Hamming distance, Manhattan norm $\ell_1$ and Euclidean norm $\ell_2$. So the lower bound for partial match implies a lower bound for $NN_{n}^d$. We restate the reductions in the following theorem.
\begin{theorem}\label{theorem-NN}
There exists a constant $C$ such that the following holds. For problem $NN_{n}^{d}$ in Hamming cube, Manhattan norm $\ell_1$ or Euclidean norm $\ell_2$, assume $d\geq{C\log n}$. If $NN_{n}^{d}$ has $(s,w,t)$-certificates where $w=d^{O(1)}$, then $t=\Omega({{d}/{\log\frac{sd}{n}}})$.
\end{theorem}
\begin{proof}
For any query pattern in $\{0,1,*\}^d$, let $k$ denote the number of $*$ in the pattern, we make the following transformations.
\begin{itemize}
\item{Hamming distance:} $0\rightarrow 01$; $1\rightarrow 10$ $*\rightarrow 11$. Here the dimension $d$ doubles but does not affect the complexity. If there exists a string that matches the pattern, then there must exist a point in the Hamming cube within distance $k$ of the query point, where $k$ denotes the number of $*$ in the pattern.
\item{$\ell_1$ and $\ell_2$ norm:} $0\rightarrow -\frac{1}{2};*\rightarrow\frac{1}{2};1\rightarrow\frac{3}{2}$. If such a string exists, then there exists a point within distance $k$ in $\ell_1$ norm and $\sqrt k$ in $\ell_2$.
\end{itemize}
Now we get the conclusion.
\end{proof}
\end{suppress}


It is well known that partial match can be reduced to 3-approximate near neighbor in $\ell_\infty$-norm by a very simple reduction~\cite{indyk2001approximate}. 
We write 3-$\ANN_n^{\lambda,d}$ for $\ANN_n^{\lambda,3,d}$.


\begin{theorem}\label{theorem-3-ANN}
Assuming $d\geq{2\log n}$, there is a $\lambda$ such that if 3-$\ANN_{n}^{\lambda,d}$ in $\ell_{\infty}$-norm has $(s,w,t)$-certificates for a $w=d^{O(1)}$, then $t=\Omega({{d}/{\log\frac{sd}{n}}})$.
\end{theorem}
\begin{proof}
We have the following model-independent reduction.
For each query pattern $x$ of partial match, we make the following transformation to each coordinate: $0\rightarrow-\frac{1}{2}$; $*\rightarrow\frac{1}{2}; 1\rightarrow\frac{3}{2}$. For a string in database the $\ell_{\infty}$-distance is $\frac{1}{2}$ if it matches pattern $x$ and $\frac{3}{2}$ if otherwise. 
\end{proof} 

\section{Lower bounds implied by lopsided set disjointness}\label{section-Unifying-Landscape}
It is observed in~\cite{patrascu11structures} that a variety of cell-probe lower bounds can be deduced from the communication complexity of one problem, the lopsided set disjointness (LSD). In~\cite{sommer2009distance}, the communication complexity of LSD is also used to prove the cell-probe lower bound for approximate distance oracle.

\ifabs{
In this section, we modify these communication-based reductions to make  them model-independent. A consequence of this is a list of certificate lower bounds as shown in Table~\ref{table-results} for: 2-Blocked-LSD, reachability oracle, 2D stabbing, 2D range counting, 4D range reporting, and approximate distance oracle.
}{
In this section, we modify these communication-based reductions to make  them model-independent. A consequence of this is a list of certificate lower bounds which match the highest known cell-probe lower bounds for respective problems, including: 2-Blocked-LSD, reachability oracle, 2D stabbing, 2D range counting, 4D range reporting, and approximate distance oracle.
}



\newcommand{\SectionLSD}{
\subsection{LSD with structures}\label{section-2-BlockedLSD}
A key idea of using LSD in reduction is to reduce from LSD with restricted  inputs.

For the purpose of reduction, the LSD problem is usually formulated as follows: the universe is $[N\cdot B]$, each query set $S\subset[N\cdot B]$ has size $N$, and there is no restriction on the size of data set $T\subseteq[N\cdot B]$. The LSD problem asks whether $S$ and $T$ are disjoint.


\begin{proposition}\label{proposition-LSD-rectangle}
For any $M\geq N$, if LSD has monochromatic 1-rectangle of size ${M\choose N}\times K$ then $K\le 2^{NB-M}$.
\end{proposition}
\begin{proof}
For any 1-rectangle of LSD, suppose the rows are indexed by $S_1,S_2,\dots,S_R$ and the columns are indexed by $T_1,T_2,\dots,T_K$. Consider the set $\mathcal{S}=\bigcup_i S_i$. Let $M=|\mathcal{S}|$. Note that $R\le \binom{M}{N}$. For any $T_i$, we have $T_i\cap \mathcal{S}=\emptyset$, so it holds that $K\leq 2^{NB-M}$.
\end{proof}

The 2-Blocked-LSD is a special case of LSD problem: the universe $[N\cdot B]$ is interpreted as $[\frac{N}{B}]\times[B]\times[B]$ and it is guaranteed that for every $x\in[\frac{N}{B}]$ and $y\in[B]$, $S$ contains a single element of the form $(x,y,*)$ and a single element of the form $(x,*,y)$.



In~\cite{patrascu11structures}, general LSD problem is reduced to 2-Blocked-LSD by communication protocols. Here we translate this reduction in the communication model to a model-independent reduction from subproblems of LSD to 2-Blocked-LSD. 







The following claim can be proved by a standard application of the probabilistic method.
\begin{Claim}[Lemma~11 in~\cite{patrascu2008data}]\label{claim-LSD-2blockedLSD}
There exists a set $\mathcal{F}$ of permutations on universe $[N\cdot B]$, where $|\mathcal{F}|=e^{2N}\cdot 2N\log B$, such that for any query set $S\subset[N\cdot B]$ of LSD, there exists a permutation $\pi\in \mathcal{F}$ for which $\pi(S)$ is an instance of 2-Blocked-LSD.
\end{Claim}

We then state our model-independent reduction as the following certificate lower bound.
\begin{theorem}\label{theorem-2Blocked-LSD}
For any constant $\delta>0$, if $2$-Blocked-LSD on universe $[\frac{N}{B}]\times[B]\times[B]$ has $(s,w,t)$-certificates, then it holds either $t=\Omega\left(\frac{NB^{1-\delta}}{w}\right)$ or $t=\Omega\left(\frac{N\log B}{\log \frac{s}{t}}\right)$.
\end{theorem}
\begin{proof}
 By Claim \ref{claim-LSD-2blockedLSD}, we know there exists a small set $\mathcal{F}$ of permutations for the universe $[N\cdot B]$ such that $|\mathcal{F}|=2^{O(N)}$ and for any input $S$ of LSD, there exists $\pi\in \mathcal F$ for which $\pi(S)$ is an instance of $2$-Blocked-LSD. By averaging principle, there exists a $\pi\in\mathcal{F}$ such that for at least $|X|/2^{O(N)}$ many sets $S$, $\pi(S)$ is an instance of 2-Blocked-LSD. Denote the set of these $S$ as $\mathcal{X}$. Restrict LSD to the domain $\mathcal{X}\times Y$ and denote this subproblem as LSD$_\mathcal{X}$.
Obviously LSD$_\mathcal{X}$ can be solved by $2$-Blocked-LSD by transforming the input with permutation $\pi$, and hence LSD$_\mathcal{X}$ has $(s,w,t)$-certificates.
For any $S\in \mathcal{X}$, there are $2^{NB-N}$ choices of $T\in Y$ such that $S\cap T=\emptyset$, so the density of 1 in LSD$_\mathcal{X}$ is at least $\frac{1}{2^N}$, thus by a standard averaging argument LSD$_\mathcal{X}$ is $(\frac{1}{2^{O(N)}}|\mathcal X|,\frac{1}{2^{O(N)}}|Y|)$-rich. Now by the richness lemma, there exists a $|{X}|/{2^{O(N+t\log\frac{s}{t})}}\times|Y|/2^{O(N+t\log\frac{s}{t}+wt)}$ 1-rectangle of LSD$_\mathcal{X}$, which is certainly a 1-rectangle of LSD. Due to Proposition~\ref{proposition-LSD-rectangle}, for any $M\geq N$, LSD has no 1-rectangle of size greater than $\binom{M}{N}\times 2^{NB-M}$, which gives us either $N+t\log \frac{s}{t}=\Omega(N\log{B}-N\log\frac{M}{N})$ or $N+tw+t\log \frac{s}{t}=\Omega(M)$. By setting $M=NB^{1-\delta}$, we prove the theorem.
\end{proof}


\subsection{Reachability oracle}\label{section-reachability-oracle}

The problem of  reachability oracle is defined as follows:
a database stores a (sparse) directed graph $G$, and reachability queries (can $u$ be reached from $v$ in $G$?) are answered.
The problem is trivially solved, even in the sense of certificates, in quadratic space by storing answers for all pairs of vertices. Solving this problem using near-linear space appears to be very hard. This is proved in~\cite{patrascu11structures} for communication protocols as well as for cell-probing schemes. We show the method in~\cite{patrascu11structures} can imply the same lower bound for data structure certificates.

\begin{theorem}\label{corollary-reachability-oracle}
If reachability oracle of $n$-vertices graphs has $(s,w,t)$-certificates for $s=\Omega(n)$, then $t=\Omega\left({\log n}/{\log{\frac{sw}{n}}}\right)$.
\end{theorem}


The lower bound is proved for a special class of graphs, namely butterfly graphs. Besides implying the general reachability oracle lower bound, the special structure of butterfly graphs is very convenient for reductions to other problems.

A butterfly graph is defined by degree $b$ and depth $d$. The graph has $d+1$ layers, each having $b^d$ vertices. The vertices on level $0$ are source vertices with 0 in-degree and the the ones on level $d$ are sinks with 0 out-degree.
On each level, each vertex can be regarded as a vector in $[b]^d$.
For each non-sink vector (vertex) on level $i$, there is an edge connecting a vector (vertex) on the $(i+1)$-th level that may differ only on  the $i$-th coordinate.
Therefore each non-sink vertex has out-degree $b$.

The problem $\RO_{n,b}$ is the reachability oracle problem defined on subgraphs of the butterfly graph uniquely specified by degree $b$ and number of non-sink vertices $n$.
For a problem $f:X\times Y\to\{0,1\}$ we define $\bigotimes^k f:X^k\times Y\to \{0,1\}$ as that $\bigotimes^k f(\bar{x},y)=\prod_{i=1}^kf(x_i,y)$ for any $\bar{x}=(x_1,x_2,\ldots,x_k)\in X^k$ and any $y\in Y$.
We further specify that in reachability oracle problem, the answer is a bit indicating the reachability, thus $\bigotimes^k\RO_{n,b}$ is well-defined.



It is discovered in~\cite{patrascu11structures} a model-independent reduction from 2-Blocked-LSD on universe $[\frac{N}{B}]\times[B]\times[B]$ to $\bigotimes^k\RO_{N,B}$ for $k=\frac{N}{d}$, where $d=\Theta(\frac{\log N}{\log B})$ is the depth of the butterfly graph.
This can be used to prove the following certificate lower bound
\begin{lemma}\label{lemma-reachability-oracle-butterfly}
If $\RO_{N,B}$ has $(s,w,t)$-certificates, then either $t=\Omega\left(\frac{d\sqrt B}{w}\right)$, or $t=\Omega\left(\frac{d\log B}{\log \frac{sd}{N}}\right)$, or $t=\Omega\left(\frac{ds}{N}\right)$, where $d=\Theta\left(\frac{\log N}{\log B}\right)$ is the depth of the butterfly graph. 
\end{lemma}
\begin{proof}
By the same way of straightforwardly combining certificates as in the proof of Lemma~\ref{lemma-direct-sum}, assuming that $\frac{N}{d}t\le s$, if $\RO_{N,B}$ has $(s,w,t)$-certificates then $\bigotimes^k\RO_{N,B}$ with $k=\frac{N}{d}$ has $(s,w,\frac{N}{d}t)$-certificates. Violating the assumption of $\frac{N}{d}t\le s$ gives us $t=\Omega\left(\frac{ds}{N}\right)$.
By the model-independent reduction in~\cite{patrascu11structures}, 2-Blocked-LSD on universe $[\frac{N}{B}]\times[B]\times[B]$ has $(s,w,\frac{N}{d}t)$-certificates.
Due to Theorem \ref{theorem-2Blocked-LSD}, for any constant $\delta>0$, either $\frac{N}{d}t=\Omega\left(\frac{NB^{1-\delta}}{w}\right)$ or $\frac{N}{d}t=\Omega\left(\frac{N\log B}{\log \frac{sd}{Nt}}\right)=\Omega\left(\frac{N\log B}{\log \frac{sd}{N}}\right)$. By setting $\delta=\frac{1}{2}$, we have either $t=\Omega(\frac{d\sqrt B}{w})$ or $t=\Omega\left(\frac{d\log B}{\log \frac{sd}{N}}\right)$. 
\end{proof}


Theorem~\ref{corollary-reachability-oracle} for general graphs is an easy consequence of this lemma.

\begin{proof}[Proof of Theorem~\ref{corollary-reachability-oracle}]
Suppose the input graphs are just those of $\RO_{N,B}$.
By Lemma~\ref{lemma-reachability-oracle-butterfly}, either $t=\Omega(\frac{d\log B}{\log \frac{sd}{N}})$, or $t=\Omega(\frac{d\sqrt B}{w})$, or $t=\Omega(\frac{ds}{N})$. Assuming $s=\Omega(n)$, the third branch becomes $t=\Omega(d)$.
Choose $B$ to satisfy $\log B=\max\{2\log w,\log \frac{sd}{N}\}=\Theta(\log\frac{sdw}{N})$.
Then we have $t=\Omega(d)$ for the first and second branches.
Since $d=\Theta({\log N/\log B})$, we have $t=\Omega({\log N}/{\log{\frac{sdw}{N}}})=\Omega({\log n}/{\log{\frac{sw}{n}}})$.
\end{proof}

Applying the model-independent reductions introduced in~\cite{patrascu11structures} from $\RO_{n,b}$ to 2D stabbing, 2D range counting, and 4D range reporting, we have the certificate lower bounds which match the highest known lower bounds for cell-probing schemes for these problems.

\begin{theorem}\label{theorem-2d-stabbing}
If 2D stabbing over $m$ rectangles has $(s,w,t)$-certificates, then $t=\Omega({\log m}/{\log{\frac{sw}{m}}})$.
\end{theorem}


\begin{theorem}\label{theorem-2d-range-counting}
If 2D range counting has $(s,w,t)$-certificates, then $t=\Omega({\log n}/{\log{\frac{sw}{n}}})$.
\end{theorem}

\begin{theorem}\label{theorem-4d-range-reporting}
If 4D range reporting has $(s,w,t)$-certificates, then $t=\Omega({\log n}/{\log{\frac{sw}{n}}})$.
\end{theorem}

\subsection{Approximate distance oracle}\label{section-distance-oracle}
For the distance oracle problem, distance queries $d_G(u,v)$ are answered for a database graph $G$.
For this fundamental problem, approximation is very important because exact solution appears to be very difficult for nontrivial settings.
Given a stretch factor $\alpha>1$, the $\alpha$-approximate distance oracle problem can be defined as:
for each queried vertex pair $(u,v)$ and a distance threshold $\tilde{d}$, the problem is required to distinguish between the two cases $d_G(u,v)\le \tilde d$ and $d_G(u,v)\ge \alpha \tilde{d}$.

We prove the following certificate lower bound for approximate distance oracle which matches the lower bound proved in~\cite{sommer2009distance} for cell-probing schemes.
\begin{theorem}\label{theorem-distance-oracle}
If $\alpha$-approximate distance oracle has $(s,w,t)$-certificates, then $t=\Omega\left(\frac{\log n}{\alpha\log(s\log n/n)}\right)$.
This holds even when the problem is restricted to sparse graphs with max degree 
$\mathrm{poly}(tw\alpha/\log n)$ for an $\alpha=o\left(\frac{\log n}{\log(w\log n)}\right)$.
\end{theorem}

We use the following notations introduced in~\cite{sommer2009distance}.
For graph $G=(V,E)$ and any two positive integers $k,\ell$, let $\mathcal{P}(G,\ell,k)$ be the set whose elements are all possible sets $P\subseteq E$ where $P$ can be written as a union of $k$ vertex-disjoint paths in $G$, each of length exactly $\ell$.
Let $g(G)$ denote the girth of graph $G$. The following claim, which is quite similar to Claim~\ref{claim-LSD-2blockedLSD}, is proved in~\cite{sommer2009distance} by the same probabilistic argument.


\begin{Claim} [Claim 13 in \cite{sommer2009distance}]\label{claim-LSD-distance-bijection}
Let $k,\ell>0$ be two integers and $N=k\ell$. Let $G=(V,E)$ be a graph with $|E|=B\cdot N$ for a positive integer $B$, and $\mathcal{P}=\mathcal{P}(G,\ell,k)$.
There exist $m$ bijections $f_1,\dots,f_m:[NB]\rightarrow E$, where $m=\ln((\mathrm{e}B)^N)\cdot \frac{(\mathrm{e}B)^N}{|\mathcal{P}|}$, such that for any $S\subseteq [NB]$ with $|S|=N$, there is a bijection $f_i$ such that $f_i(S)\in\mathcal{P}(G,\ell,k)$.
\end{Claim}

Consider the problem of $\alpha$-approximate distance oracle for base-graph $G$, in which
the $\alpha$-approximate distance queries are answered only for spanning subgraphs of $G$. The following lemma is the certificate version of a key theorem in~\cite{sommer2009distance}.



 \begin{lemma}\label{lemma-distance-oracle-general}
There exists a universal constant $C$ such that the following holds.
Let $G=(V,E)$ be a graph, such that $\alpha$-approximate distance oracle for the base-graph $G$ has $(s,w,t)$-certificates.
Let $k,\ell$ be two positive integers, such that $\ell<\frac{g(G)}{\alpha+1}$. Assume $|E|\geq k\ell(2tw/\ell)^{1/C}$. Then
 \begin{equation*}
s\geq \frac{k}{e}\left(\frac{|\mathcal{P}(G,\ell,k)|^{1/{k\ell}}}{e(|E|/k\ell)^{1-C}}\right)^{\frac{\ell}{t}}\left(e|E|\right)^{-\frac{1}{tk}}
\end{equation*}
 \end{lemma}
 \begin{proof}
Suppose $N=k\ell$ and $B=|E|/N$.
Consider the LSD problem LSD$:X\times Y\to\{0,1\}$ defined on universe $[N\cdot B]$ such that each query set $S\subset[N\cdot B]$ is of size $|S|=N$ and each dataset $T\subseteq [N\cdot B]$ is of arbitrary size.
By Claim~\ref{claim-LSD-distance-bijection}, there exists $m$ bijections, $f_1,\dots,f_m:[NB]\rightarrow E$, where $m=\ln((eB)^N)\cdot \frac{(eB)^N}{|\mathcal{P}|}$ where $\mathcal{P}=\mathcal{P}(G,\ell,k)$, such that for any $S\subseteq [NB]$ with $|S|=N$, there exists a bijection $f_i$ such that $f_i(S)\in\mathcal{P}(G,\ell,k)$. By averaging principle, there exists an $f_i$ such that for at least $|X|/m$ many sets $S$, it holds that $f_i(S)\in\mathcal{P}(G,l,k)$. Denote the set of such $S$ as $\mathcal{X}$.
Restrict LSD to the domain $\mathcal{X}\times Y$ and denote this subproblem as LSD$_\mathcal{X}$.
Next we prove LSD$_\mathcal{X}$ can be solved by a composition of $\alpha$-approximate distance oracles.

Let $f_i$ be the bijection such that $f_i(S)\in\mathcal{P}(G,\ell,k)$ for all  $S\in \mathcal{X}$.
For any $S\in \mathcal{X}, T\subseteq[N\cdot B]$, an instance for approximate distance oracle for the base graph $G=(V,E)$ is constructed  as follows. The database graph for distance oracle  is the spanning subgraph $G'=(V,E')$ where $E'=E\setminus f_i(T)$.
Due to the property of bijection $f_i$, it holds that $P=f_i(S)$ contains $k$ vertex-disjoint paths $p_1,p_2,\ldots, p_k$, each of length $\ell$. Let $(u_1,v_1),\dots,(u_k,v_k)$ denote the pairs of end-vertices of these paths.
Since $f_i$ is a bijection, the disjointness of $S$ and $T$ translates to the disjointness of $f_i(S)$ and $f_i(T)$, i.e.~all these $k$ vertex-disjoint paths are intact by removing edges in $f_i(T)$ from the graph $G$.


\newcommand{\ADO}{\alpha\mbox{-}\mathrm{Dist}}

Consider the $\alpha$-approximate distance oracle problem $\ADO_G$ for the base-graph $G$. 
We then observe that LSD$_\mathcal{X}$ can be solved by the problem $\bigotimes^k \ADO_G$ of answering $k$ parallel approximate distance queries, where $\bigotimes^k f$ of a problem $f$ is as defined in last section.
Consider the $k$ vertex pairs $(u_i,v_i), i=1,2,\ldots, k$ connected by vertex-disjoint paths $p_i$ constructed above. We have $d_{G}(u_i,v_i)=\ell$ for every $1\le i\le k$.
For $\ADO_G$, if all edges in $p_i$ are in $E'$, then $d_{G'}(u_i,v_i)\leq \ell$, so $\ADO_G((u_i,v_i, \ell), G')$ will return ``yes'', and if there is an edge in $p_i$ is not in $E'$, since graph $G$ has girth $g(G)>(\alpha+1)\ell$,
we must have $d_{G'}(u_i,v_i)\ge g(G)-\ell>\alpha \ell$, so $\ADO_G((u_i,v_i, \ell), G')$ will return ``no''.
By above discussion, if $\ADO_G((u_i,v_i, \ell),G')$ returns ``yes'' for all $k$ queries then it must hold $S\cap T=\emptyset$, and if $\ADO_G((u_i,v_i, \ell),G')$ returns ``no'' for  some $i$, then $S\cap T\ne \emptyset$, i.e.~we have a model-independent reduction from LSD$_\mathcal{X}$ to $\bigotimes^k \ADO_G$.

If the $\alpha$-approximate distance oracle problem $\ADO_G$ has $(s,w,t)$-certificates, then by directly combining $k$ certificates for $k$ parallel queries, the problem $\bigotimes^k \ADO_G$ has $(s,w,kt)$-certificates, and hence LSD$_\mathcal{X}$ has $(s,w,kt)$-certificates.
For every $S\in\mathcal{X}$, there are $2^{NB-N}$ many $T$ disjoint with $S$, so the density of LSD$_\mathcal{X}$ is at least $2^{-N}$. By a standard averaging argument, this means LSD$_\mathcal{X}$ is $(\frac{1}{2^{N+1}}|\mathcal{X}|,\frac{1}{2^{N+1}}|Y|)$-rich.
By Lemma~\ref{lemma-richness-cert}, there exist universal constants $C_1,C_2>0$ such that LSD$_\mathcal{X}$ has monochromatic 1-rectangle of size $|\mathcal{X}|/2^{O(N+kt\log \frac{s}{kt})}\times|Y|/2^{O(N+kt\log \frac{s}{kt}+ktw)}$, which is also 1-rectangle of LSD.
Note that $|\mathcal{X}|\geq|X|/m=|X|/\ln((eB)^N)\cdot \frac{(eB)^N}{|\mathcal{P}|}$, so the rectangle is of size at least
\begin{equation*}
 |X|/2^{O(N\log(eB)+\log(eBN)-\log(|\mathcal{P}|)+kt\log \frac{s}{kt})}\times|Y|/2^{O(N+kt\log \frac{s}{kt}+ktw)},
\end{equation*}
where the big-O notations hide only universal constants. And for LSD, $|X|={NB\choose N}$ and $|Y|=2^{NB}$.
Due to Proposition~\ref{proposition-LSD-rectangle}, for any $M\geq N$, LSD has no 1-rectangle of size greater than $\binom{M}{N}\times 2^{NB-M}$.
By a calculation, there exist a universal constant $C>0$ such that by considering an $M=\Theta(NB^C)$, we have either $tk\log(s/k)+N\log(eB)+\log(eBN)-\log(|\mathcal{P}|)\geq CN\log B$ or $ktw\geq NB^{C}$.
Since the lemma assumes $|E|\ge k\ell(2tw/\ell)^{1/C}$, we have  $B\geq (2tw/\ell)^{1/C}$, thus $NB^C\geq k\ell\cdot2tw/\ell=2ktw$. The second branch can never be satisfied. And by a calculation, the first branch gives us the bound of the lemma.
 \end{proof}


The following graph-theoretical theorem is proved in~\cite{sommer2009distance}.
\begin{theorem}[combining Lemma 14, Theorem 9, 17, and 18 of~\cite{sommer2009distance}]\label{theorem-final-proof}
Let $n$ be sufficiently large.
For any constant $C>0$, any $t=t(n)$ and any $\alpha=\alpha(n), w=w(n)$ satisfying $w=n^{o(1)}$ and $\alpha=o\left(\frac{\log n}{\log(w\log n)}\right)$. There exist $r=r(n)$ and $r$-regular graph $G=G_n$ of $n$ vertices, such that
\begin{itemize}
\item $r\geq (4tw\alpha/g(G))^{1/C}$;
\item $2\alpha\leq g(G)\leq\log n$;
\item $|\mathcal{P}(G,\ell,k)|^{1/k\ell}=\Omega(r)$ for $\ell=\left\lfloor g(G)/2\alpha \right\rfloor$ and $k=n/20\ell$;
\item $r^{g(G)}=n^{\Omega(1)}$.
\end{itemize}
\end{theorem}
Now we prove Theorem~\ref{theorem-distance-oracle} by applying Lemma~\ref{lemma-distance-oracle-general} to the sequence of regular graphs $G_n$ constructed in Theorem~\ref{theorem-final-proof}. Note that in $G_n$, we have $|E|=n\cdot r/2=10k\ell\cdot r\geq 10k\ell\cdot(2tw/\ell)^{1/C}\geq k\ell(2tw/\ell)^{1/C}$, so the assumption of Lemma~\ref{lemma-distance-oracle-general} is satisfied. On the other hand, we have $|\mathcal{P}(G,\ell,k)|^{1/k\ell}=\Omega(r)$ and $e(|E|/k\ell)^{1-C}=\Theta(r^{1-C})$. Since $|E|\leq n^2 \leq (\ell k)^4\leq k^8$, we have $(e|E|)^{1/tk}=\Theta(1)$. And it holds that $k=\frac{n}{20\ell}=\Omega(\frac{n}{\log n})$. Ignoring constant factors, the bound in Lemma~\ref{lemma-distance-oracle-general} implies:
\begin{equation*}
s\geq \frac{n}{\log n}\left(\frac{r}{r^{1-C}}\right)^{\Omega(\ell/t)}=\frac{n}{\log n}r^{\Omega(\ell/t)}=\frac{n}{\log n}r^{\Omega(g(G)/\alpha t)}=\frac{n^{\Omega(1+1/\alpha t)}}{\log n}
\end{equation*}
Translating this to a lower bound of $t$, we have $t=\Omega\left(\frac{\log n}{\alpha \log(s\log n/n)}\right)$.
}

\ifabs{
The rest of this section can be found in Appendix~\ref{appendix-LSD}.
}{



\SectionLSD
}

\paragraph{Acknowledgment.}
We are deeply grateful to Kasper Green Larsen for helpful discussions about lower bound techniques in cell-probe model.


\ifabs{
\section*{Appendix}
\appendix
\input{appendix}
}
{}

\end{document}